\newcount\Comments  
\documentclass[a4paper,UKenglish,cleveref, autoref, thm-restate]{lipics-v2021}

\usepackage{graphicx}
\usepackage{textcomp, gensymb}
\usepackage{amsmath, amssymb}

\usepackage{enumitem}
\usepackage[noend, ruled, lined, linesnumbered, commentsnumbered,vlined, longend]{algorithm2e}
\usepackage{xcolor}
\usepackage[normalem]{ulem}
\usepackage{soul}

\SetCommentSty{mycommfont}
\definecolor{darkgreen}{rgb}{0,0.6,0}
\definecolor{purple}{rgb}{1,0,1}

\definecolor{darkgreen}{rgb}{0,0.6,0}
\definecolor{purple}{rgb}{1,0,1}

\newcommand{\kibitz}[2]{\ifnum\Comments=1{\color{#1}{#2}}\fi}

\title{On the Exponential Growth of Geometric Shapes} 

\titlerunning{On the Exponential Growth of Geometric Shapes} 

\author{Nada Almalki}{Department of Computer Science, University of Liverpool, UK}{n.almalki@liverpool.ac.uk}{}{}

\author{Siddharth Gupta}{Department of Computer Science \& Information Systems, BITS Pilani Goa Campus, India}{siddharthg@goa.bits-pilani.ac.in}{https://orcid.org/0000-0003-4671-9822}{}

\author{Othon Michail}{Department of Computer Science, University of Liverpool, UK}{othon.michail@liverpool.ac.uk}{https://orcid.org/0000-0002-6234-3960}{}

\authorrunning{N. Almalki, S. Gupta, and O. Michail}

\Copyright{Almalki, Gupta, Michail}

\usepackage{hyperref}
\EventEditors{}
\EventNoEds{2}
\EventLongTitle{}
\EventShortTitle{}
\EventAcronym{}
\EventYear{}
\EventDate{}
\EventLocation{}
\EventLogo{}
\SeriesVolume{}
\ArticleNo{}
\nolinenumbers

\begin{document}

\maketitle
\begin{abstract}
In this paper, we explore how geometric structures can be grown exponentially fast. The studied processes start from an initial shape and apply a sequence of centralized growth operations to grow other shapes. We focus on the case where the initial shape is just a single node. A technical challenge in growing shapes that fast is the need to avoid collisions caused when the shape breaks, stretches, or self-intersects. We identify a parameter $k$, representing the number of turning points within specific parts of a shape. We prove that, if edges can only be formed when generating new nodes and cannot be deleted, trees having $O(k)$ turning points on every root-to-leaf path can be grown in $O(k\log n)$ time steps and spirals with $O(\log n)$ turning points can be grown in $O(\log n)$ time steps, $n$ being the size of the final shape. For this case, we also show that the maximum number of turning points in a root-to-leaf path of a tree is a lower bound on the number of time steps to grow the tree and that there exists a class of paths such that any path in the class with $\Omega(k)$ turning points requires $\Omega(k\log k)$ time steps to be grown. If nodes can additionally be connected as soon as they become adjacent, we prove that if a shape $S$ has a spanning tree with $O(k)$ turning points on every root-to-leaf path, then the adjacency closure of $S$ can be grown in $O(k \log n)$ time steps. In the strongest model that we study, where edges can be deleted and neighbors can be handed over to newly generated nodes, we obtain a universal algorithm: for any shape $S$ it gives a process that grows $S$ from a single node exponentially fast.

\end{abstract}

\section{Introduction}\label{sec:intro}

\subsection{Scenario: Planet Exploration by Self-Replicating Robotic Modules}\label{subsec:scenario}

Suppose the modules of a robotic system are designed in such a way that, upon receiving a special replication signal, they can generate copies of themselves in the local space surrounding them, thus increasing the size and changing the structure of the robot. We assume that a centralized algorithm controls which modules will replicate each time and in which direction, doing so by generating appropriate replication signals and transmitting them to the relevant modules. Through this mechanism, the algorithm can control the growth dynamics of the robotic structure, resembling the way the genome controls embryonic development in multicellular organisms.

Now imagine that a planet exploration mission requires the deployment of $r$ worm-like robots, each robot consisting of $n$ modules connected in a chain. Both $r$ and $n$ are assumed to be large positive integers. Under these assumptions, the following would be a cost- and time-efficient strategy to transport and deploy the robots. Only a few individual modules are transported, which, once on the planet, self-replicate successively until $r$ individual modules are generated. The following algorithm is then applied in parallel on the $r$ modules, each assumed to be a trivial initial chain of unit length. Operating in discrete time steps, the algorithm keeps sending a ``(right, connected)'' replication signal to every module of the chain. Upon receiving this signal, every module $u$ which is not the rightmost module of the chain generates a new module $u'$ between itself and its right neighbor $v$, $u'$ becomes connected to both $u$ and $v$, and the previous connection between $u$ and $v$ is dropped. The rightmost module $u$ of the chain generates a new module $u'$ to its right and connects to it. As the length of every chain doubles in each time step, after $\log n$ time steps the algorithm will have grown an $n$-module robot out of each of the $r$ individual modules.

In this work, we explore two interrelated questions that generalize this scenario: \emph{``What are the structural properties associated with exponential growth of geometric shapes?''} and \emph{``How can some of these properties be exploited and others avoided in order to design algorithms that can grow desired shapes exponentially fast?''}

Consider, for example, the question asking if a robot, whose $n$ modules ---called \emph{nodes} hereafter--- are arranged in a tree shape $T$, can be grown from a single node $u_0$ exponentially fast. One possible approach is to use breadth-first search on the line segments of $T$. Starting from $u_0$, the algorithm proceeds in phases of increasing line-segment distance, measured by the number of maximal line segments from $u_0$. In each phase $i$, the algorithm grows in parallel all maximal line segments at a line-segment distance $i$ from $u_0$. Each line segment can be grown exponentially fast by the chain algorithm of the planet exploration scenario. If $T$ is sparse (informally meaning its branches are sufficiently far from each other) and has $O(k)$ line segments on every root-to-leaf path, then this algorithm would grow $T$ in $O(k\log n)$ time steps, which gives a logarithmic number of time steps if $k$ is a constant but can be as slow as linear for $k=\Omega(n/\log n)$. A less obvious observation is that, when $T$ is not sparse, we need to specify how the efficiency or even the feasibility of growing a branch depends on the presence of nearby branches. 

These observations prompt further questions that we explore in the present paper, such as: Is there an algorithmic approach that can grow any tree or any general shape of $n$ nodes in $O(\log n)$ time steps? What other classes of shapes can be identified for which (poly)logarithmic upper bounds or super-logarithmic lower bounds can be proved? What are the relevant parameters of the geometry of a shape that affect the complexity of growth? What are the right modeling assumptions to represent, for example, how parallel local growth affects global growth, the different types of collisions that should be avoided, or properties related to the connectivity of the shape?        

\subsection{Motivation and Related Work}\label{subsec:related-work}

During the early stages of an organism's development, cells undergo exponential growth, beginning with a single cell and successively doubling their total number. Mathematical models of human embryonic growth have been proposed~\cite{luecke1999mathematical}. 
Though our model takes inspiration from natural growth processes, it also shares features with existing theoretical models of computation and robotics. 

The abstract Tile Assembly Model~\cite{winfree1996computational,winfree1998algorithmic} is a mathematical model of self-assembly. 
The process starts from an initial 
assembly of tiles  
and the system grows, much like crystals in nature, by passive attachment of new tiles, which bind on existing tiles through glues on their edges. \emph{Passive} means that interactions between the tiles are controlled by the environment. In practice,
tiles are typically DNA molecules; see~\cite{doty2012theory,patitz2014introduction} for surveys of work on the theory of algorithmic self-assembly. Though growth is a defining property of both our model and self-assembly models, in the majority of self-assembly models growth is through passive attachment on the external layer of the formed structure, including any internal cavities. In contrast, algorithms in our model can actively control the structure's growth and can do so without an \emph{a priori} limitation on where to apply the growth operations. As a result of this, growth in passive self-assembly is a relatively slow process both in theory and in practice, whereas the algorithms we develop use a number of time steps which is typically sub-linear and often (poly)logarithmic in the size of the final structure.

An example of a self-assembly model incorporating active molecular dynamics is the \emph{Nubot} model~\cite{woods2013active}. The authors explore how to grow connected two-dimensional geometric shapes and patterns in time (poly)logarithmic in their size, an objective that is broadly similar to ours.  
For example, they show how to grow a chain (like the one in the planet exploration scenario) or a square of monomers in time and number of monomer states which is logarithmic in the total number of monomers in the final structure. They use these constructions as a basis for constructing arbitrary shapes exponentially fast, combining further growth and other forms of reconfiguration. A difference between our model and \cite{woods2013active} is that our processes are only allowed to update instances through growth.
As is also the case in \cite{woods2013active}, most of our algorithms use the fast process of growing a chain as a sub-routine.

\emph{Programmable matter} refers to systems composed of small and uniform robotic modules. The collection of modules can dynamically change its physical properties, such as its shape or density. Recently, there has been growing interest in studying the algorithmic foundations of such systems, focusing on their ability to alter their shape through local reconfiguration~\cite{akitaya2021universal, almethen2020pushing, derakhshandeh2016universal, michail2019transformation}. 
Control can be either centralized or decentralized and the focus is typically on the feasibility of a given reconfiguration task, as most of these models share a natural quadratic lower bound on worst-case moves based on reconfiguration distance (translating to a linear lower bound if parallel moves are allowed) \cite{michail2019transformation}. An exception is a recent line of work on models for fast parallel reconfiguration \cite{feldmann2022coordinating,padalkin2023shape,daymude2023canonical}, extending the Amoebot model of Derakhshandeh \emph{et al.}  \cite{derakhshandeh2014amoebot}.
The growth processes studied in the present paper could serve as a way to deploy programmable matter fast, either in its exact initial configuration or in a rough version of it that can be then refined through other types of operations.

An assumption of our model is that individual operations have \emph{linear strength}, meaning that they have enough power to move any part of the structure. This has been a common simplifying assumption in the relevant literature and can sometimes be dropped, e.g., when more than one operation can be applied in parallel to collaboratively provide the required ``force''. Examples of other linear-strength models are \cite{aloupis2008reconfiguration,woods2013active,almethen2020pushing}.

Other studies close to our work are \cite{almalki2024geometric}, which considered simpler forms of the growth models and operations studied in the present paper, \cite{gupta2023collision}, which explored the complexity of deciding whether collisions will occur or can be avoided for given set of growth and contraction operations, and \cite{mertzios2022complexity}, where the authors studied related growth problems on abstract (not necessarily geometric) graphs. There are important differences between the models of \cite{almalki2024geometric} and the ones studied in the present paper. Growth operations in \cite{almalki2024geometric} are conditioned to affect specific parts of the shape (e.g., whole columns) and are defined in a way which ensures in advance that no collision can happen. In contrast, the growth operations in the present paper can be applied to any parts of the shape and it is the responsibility of the algorithm designer to ensure that no structural violation occurs. Our model also has some relevance to von Neumann's concept of self-replicating machines.

\subsection{Contribution}\label{subsec:contribution}

Our aim is broadly to understand how geometric structures can be grown exponentially fast. Though this question is also relevant to distributed algorithms,
our present study is \emph{centralized}; both cases are largely unexplored and we naturally aim to understand the centralized case first. The distributed case remains a direction for future research. Nevertheless, our centralized lower bounds also apply to distributed algorithms and it might be possible to translate our centralized algorithms into ---potentially less efficient--- distributed algorithms.

In this paper, we represent geometric structures as shapes drawn on a two-dimensional  
square grid. Though three-dimensional 
space and other coordinate systems (including continuous) can be more suitable for some applications, the considered model is theoretically simpler and the techniques developed for it are likely to generalize. A shape is a geometric graph that is formed by nodes occupying distinct points on the grid and edges between pairs of adjacent nodes (i.e., nodes at a unit orthogonal distance from each other on the grid). We naturally restrict attention to connected shapes.

Intuitively, a \emph{growth process} is a sequence of sets of growth operations that are applied starting from an initial shape $S_0$ ---often a single node--- to grow a final shape $S$ of $n$ nodes. We restrict attention to discrete time, which we measure in \emph{time steps}. In all problems we study, we want to determine a bound $\tau$ such that for all shapes $S_0$ and $S$ from given initial and final classes of shapes, respectively, there is a growth process $\sigma$ that grows $S$ from $S_0$ in $\tau$ time steps. For an upper bound to establish exponential growth, $\tau$ should be at most (poly)logarithmic in the size of $S$. Non-trivial lower bounds should be super-logarithmic in the size of $S$.

When the initial shape is a single node, it is probably plausible to assume that there might be a straightforward solution to this problem. One could try to grow a small compressed version $T'$ of a spanning tree $T$ of shape $S$, decompress $T'$ through parallel growth to get $T$ exponentially fast, and add any missing edges in order to get $S$. However, under reasonable assumptions, this generic approach fails to be exponentially fast in the worst case. One reason for this is that there might be no meaningful compression of $S$. The second reason is that even when there is a satisfactory compression, fast decompression might not be possible if self-intersection must be avoided. We call all structural violations of the shape \emph{collisions} and are only interested in growth processes that are free from collisions. 

In Section~\ref{sec:models-preliminaries}, we define a general model of growth and distinguish some of its variants. One distinction is drawn between those processes that can only form edges when generating new nodes (\emph{connectivity graph}) and those that can additionally connect nodes as soon as they become adjacent (\emph{adjacency graph}). Another is based on whether structural violations of cycles are to be avoided by preserving the affected cycles (\emph{cycle-preserving} processes) or by breaking them (\emph{cycle-breaking} processes). 
In Section~\ref{sec:connectivity-graph-model}, we study the connectivity graph model. 
We identify a parameter $k$, representing the number of turning points within specific parts of a shape, and show 
that, when starting from a single node, trees having $O(k)$ turning points on every root-to-leaf path can be grown in $O(k \log n)$ time steps through BFS on line segments and spirals with $O(\log n)$ turning points can be grown in $O(\log n)$ time steps through a pipelined version of BFS. We also establish two lower bounds for trees and paths. For trees, we show that  
the maximum number of turning points in a root-to-leaf path of the tree is a lower bound on the number of time steps to grow the tree from a single node. For paths, we show that there exists a class of paths such that any path in the class with $\Omega(k)$ turning points requires $\Omega(k\log k)$ time steps to be grown from a single node. This shows that pipelined BFS, or any other approach, cannot give a logarithmic upper bound for all paths with a logarithmic number of turning points, because there are paths in the class with $\Theta(\log n)$ turning points. Though an $\Omega(n)$ lower bound for staircase shapes with $\Theta(n)$ steps from \cite{almalki2024geometric} applies to the model studied in Section~\ref{sec:connectivity-graph-model}, it would not exclude (poly)logarithmic upper bounds for paths with (poly)logarithmic many turning points. This is because that lower bound does not generalize to staircases with (poly)logarithmically many turning points: in fact, a staircase with a logarithmic number of turning points can be grown in a logarithmic number of time steps by first growing the turning points and then bringing the staircase to its final length by growing its line segments in parallel.

In Section~\ref{sec:adj-graph-model}, we discuss the cycle-preserving and cycle-breaking types of processes within the adjacency graph model, in which every pair of adjacent nodes is also connected in the shape. Note that this distinction is not meaningful in the connectivity graph model for processes starting from a single node, because the shapes grown by these processes are acyclic. 
For cycle-preserving processes, we prove that if a shape $S$ has a spanning tree with $O(k)$ turning points on every root-to-leaf path, then the adjacency closure of $S$ can be grown from a single node in $O(k \log n)$ time steps. 
For cycle-breaking processes with the additional assumption that neighbors can be handed over to newly generated nodes (\emph{neighbor handover}), we show that there exists a universal algorithm: for any shape $S$ it gives a process that grows $S$ from a single node exponentially fast. 
Even though the combination of cycle-breaking and neighbor handover is a rather strong one, it shows that there is at least one reasonable model variant which can give a universal algorithm for exponential growth. To what extent neighbor handover can help increase the class of shapes that can be grown exponentially fast by cycle-preserving processes is an open question. In Section~\ref{sec:comparative-overview}, we show a comparative overview of the adjacency and connectivity graph models, highlighting the shape classes constructible by each.
In Section~\ref{sec:conclusion}, 
we discuss open problems raised by our model and its variants.

\section{Models and Problem}\label{sec:models-preliminaries}

We consider a two-dimensional square grid, each point of which  
is identified by its $x\geq 0$ and $y\geq 0$ integer coordinates, $x$ indicating the column and $y$ the row. A \emph{shape} is defined as a graph $S=(V, E)$ drawn on the grid. $V$ is a set of $n$ nodes, where each node $u$ occupies
a distinct point $(u_x,u_y)$ of the grid. $E\subseteq \{uv\;|\; u,v\in V \text{ and } u,v \text{ are adjacent}\}$ is a set of edges between pairs of adjacent nodes, where two nodes $u=(u_x, u_y)$ and $v=(v_x, v_y)$ are \emph{adjacent} if $u_x\in\{v_x-1,v_x+1\}$ and $u_y=v_y$ or $u_y\in\{v_y-1,v_y+1\}$ and $u_x=v_x$, that is, their orthogonal distance on the grid is one. 
We illustrate nodes as small circles drawn on the points they occupy; however, our results hold for any geometry of individual nodes that does not trivially make nearby nodes intersect.  
A shape is \emph{connected} if the graph that defines it is a connected graph. Throughout the paper, we restrict attention to connected shapes. We define the \emph{adjacency closure} of a shape $S=(V,E)$ as the shape $AC(S)=(V, E')$, where $E' = E \cup \{uv\;|\; u, v \in V$, $u, v$ are adjacent, and $uv \notin E\}$.

A \emph{growth operation} (also called \emph{doubling} \cite{almalki2024geometric} or \emph{expansion} \cite{gupta2023collision}) applied on a node $u$ generates a new node in one of the points adjacent to $u$ and possibly translates some part of the shape. 
One or more growth operations applied in parallel to nodes of a shape $S$ either cause a \emph{collision} or yield a new shape $S'$. There are two types of collisions: \emph{node collisions} and \emph{cycle collisions}. When describing the outcome of growth operations, we will be assuming that 
there is an \emph{anchor} node 
which is stationary and other nodes move relative to it.
This is without loss of generality under the assumption that the constructed shapes are equivalent up to translations. 
 
Temporarily disregarding collisions, we begin by defining the effect of a single growth operation on a tree shape; we will add collisions to the definition of the general case, of one or more operations applied in parallel. Let $T=(V,E)$ be a tree and $u_0\in V$ its anchor. We set $u_0$ to be the root of $T$. A single growth operation is applied on a node $u\in V$ toward a point $(x,y)$ adjacent to $u$. 
There are two cases: (i) there is no edge between $u$ and $(x,y)$, (ii) $(x,y)$ is occupied by a node $v$ and $uv\in E$. We first define the effect in each of these cases when \emph{neighbor handover} is not allowed. In case~(i), the growth operation generates a node $u'$ at point $(x,y)$ and connects it to $u$. In case~(ii), assume without loss of generality that $u$ is closer to $u_0$ in $T$ than $v$.
Let $T(v)$  denote the subtree of $T$ rooted at node $v$. Then, the operation generates a node $u'$ between $u$ and $v$, connected to both, which translates $T(v)$ by one unit away from $u$ along the axis parallel to~$uv$. After this, $u'$ occupies $(x,y)$ and $uv$ has been replaced by $\{uu',u'v\}$. If neighbor handover is allowed, then any neighbor $w$ of $u$ perpendicular to $uu'$ can be handed over to $u'$. This happens by a unit translation of $T(w)$ or $T(u)$ along the axis parallel to~$uu'$, depending on which of $u,w$, respectively, is closer to $u_0$ in $T$.

Let $Q$ be a set of operations to be applied \emph{in parallel} to a connected shape $S$, each operation on a distinct pair of nodes or a node and an unoccupied point. 
We assume that all operations in such a set 
are applied \emph{concurrently}, have the same \emph{constant execution speed}, and their \emph{duration} is equal to one \emph{time step}. Let again $T=(V,E)$ be a tree, $u_0\in V$ its anchor, and set $u_0$ to be the root of $T$. We want to determine the displacement of every $v\in V\setminus\{u_0\}$ and of every newly generated $v$ due to the parallel application of the operations in $Q$. As $u_0$ is stationary and each operation translates a subtree, only the operations on the unique $u_0v$ path contribute to $v$'s displacement.
In particular, any such operation contributes one of the unit vectors $\langle -1,0\rangle, \langle 0,-1\rangle, \langle +1,0\rangle, \langle 0,+1\rangle$ to the motion vector $\vec{v}$ of~$v$. We can use the set of motion vectors to determine whether the trajectories of any two nodes will collide at any point. This type of collision is called a \emph{node collision} (see Figure~\ref{fig:node-collision}).

Let $S$ be any connected shape with at least one cycle and any node $u_0$ be its anchor. Then, 
a set of parallel operations $Q$ on $S$ either causes a \emph{cycle collision} or its effect is essentially equivalent to the application of $Q$ on any spanning tree of $S$ rooted at $u_0$. 
Let $u$, $v$ be any two nodes on a cycle. If $p_1$ and $p_2$ are the two paths between $u$ and $v$ of the cycle, then $\vec{v}_{p_1}=\vec{v}_{p_2}$ must hold:
the displacement vectors along the paths $p_1$ and $p_2$ are equal. 
Otherwise, we cannot maintain all nodes or edges of the cycle. Such a violation is called a \emph{cycle collision} and an example is shown in Figure~\ref{fig:cycle-collision}. A set of operations is said to be \emph{collision free} if it does not cause any node or cycle collisions.

\begin{figure}[ht]
\centering 
\includegraphics[scale=0.65]{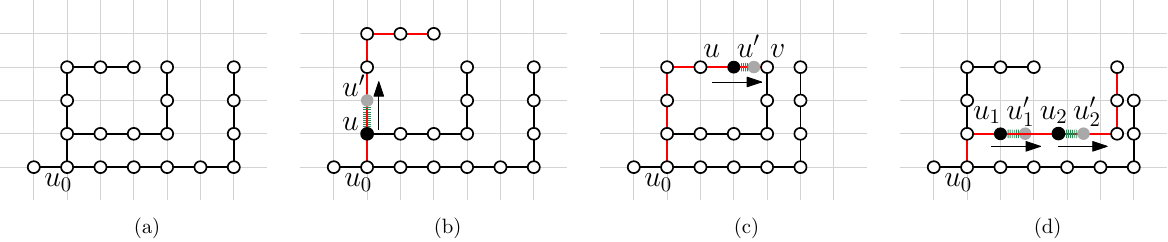}
\captionsetup{justification=justified, singlelinecheck=false}
\caption{(a) Initial tree $T$. (b) The tree $T'$ after a growth operation on node $u$ moves north to generate a new node $u'$ without any collision. (c) Illustration of a node collision scenario: the $T'$ here is a result of a growth operation applied on node $u$ toward the east, where a node $v$ already exists and $uv\notin E$. The newly generated node $u'$ occupies the same position as $v$, leading to a node collision. (d) Another scenario of a node collision: two nodes $u_1$ and $u_2$ simultaneously grow in the east direction and generate $u_1'$ and $u_2'$, though $u_1'$ and $u_2'$ do not collide directly, their growth pushes their branch into an adjacent branch, leading to a collision.}
\label{fig:node-collision}
\end{figure}

\begin{figure}[ht]
\centering 
\includegraphics[scale=0.65]{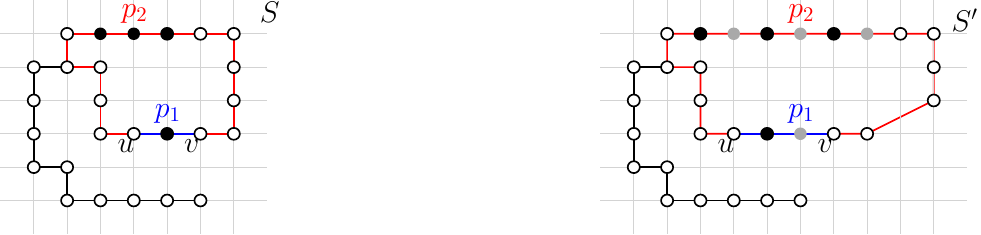}
\captionsetup{justification=justified, singlelinecheck=false}
\caption{An example of a cycle collision within the shape $S$ due to unequal displacement vectors along the two paths $p_1$ and $p_2$, thus, $\vec{v}_{p_1}\neq \vec{v}_{p_2}$. In particular, the number of generated nodes (gray nodes) along $p_2$ is greater than that along $p_1$. This difference in the number of generated nodes leads to a collision within the cycle, indicating an irregularity in the shape $S'$.
}
\label{fig:cycle-collision}
\end{figure}

A growth process $\sigma$ starts from an initial shape $S_0$ ---often a single node--- and,
in each time step $t\geq 1$, applies a set of parallel growth operations ---possibly a single operation--- on the current shape $S_{t-1}$ to give the next shape $S_t$, until a final shape $S$ is reached at a time step $t_f$. In this case, we say that \emph{$\sigma$ grows $S$ from $S_0$ in $t_f$ time steps}. We also assume that parallel operations have the same cardinal direction and that a node gets at most one operation per time step. These assumptions simplify the description of algorithms and can be easily dropped.

A distinction can be drawn between processes that can only form edges when generating new nodes and processes that can additionally connect nodes as soon as they become adjacent. We further distinguish between processes that avoid cycle collisions by preserving cycles and those that can also do so by breaking them. These are defined as follows.

\begin{definition} \label{def:models}
Let $S_t^b$ and $S_t^e$ denote the shapes formed by the beginning and by the end of time step $t$, respectively, and assume that $S_1^b=S_0$. A \emph{cycle-preserving} growth process applies a collision free set of parallel growth operations $Q_t$ to $S_t^b$, for all time steps $t\geq 1$.
A \emph{cycle-breaking} growth process additionally removes a ---possibly empty--- subset of the edges of $S_t^b$, whose removal does not disconnect the shape, before applying $Q_t$ to it. If \emph{neighbor handover} is allowed, growth of a node $u$ generating a new node $u'$ in direction $d$ can hand any neighbor $w$ of $u$ perpendicular to $d$ over to $u'$. 
In the \emph{connectivity graph} model, for all $t\geq 1$, $S_{t+1}^b=S_t^e$ holds. 
In the \emph{adjacency graph} model, for all $t\geq 1$, $S_{t+1}^b=AC(S_t^e)$ holds. 
\end{definition}

Intuitively, the additional assumption in the adjacency graph model is that, at the end of every time step, the graph model updates the shape by connecting all adjacent nodes that are not connected. Combining the adjacency graph model with cycle-breaking processes captures the less extreme case, in which the process can choose any spanning connected sub-shape of the adjacency closure. 

\begin{proposition}\label{prop:models-diff}
For the models of Definition \ref{def:models}, the following properties hold:
\begin{enumerate}
\item Under the connectivity 
model, the growth processes never increase the number of cycles.
\item Under the connectivity 
model, if $S_0$ is a single node, the processes can only grow trees. 
\item The cycle-preserving process never decreases the number of cycles.

\item Under the connectivity model, the cycle-preserving process preserves the number of cycles.
\end{enumerate}
\end{proposition}
\begin{proof}
    Property (2) is a special case of (1). Property (4) follows by taking (1) and (3) together. So, it is sufficient to prove properties (1) and (3). We first prove this \emph{without neighbor handover}. In that case, the cycle-preserving process cannot remove any edges, and neither do the graph models; thus, property (3) holds. Property (1) follows by observing that, without neighbor handover, the growth processes can only add leaves or increase the length of existing line segments and that the connectivity graph model does not modify any edges. We now show that these remain true when \emph{neighbor handover is allowed}. Let $u$ be a node on which a growth operation is applied, and $u_N,u_E,u_S,u_W$ its up to 4 neighbors in the respective cardinal directions. Without loss of generality, let $u^{\prime}_E$ be the node generated by the operation in the east direction. The nodes that can be handed over from $u$ to $u^{\prime}_E$ are $u_N$ and $u_S$. If we show that the number of cycles is invariant of handover for both types of processes, then propositions (1) and (3) will follow. It is sufficient to consider those cycles that before applying the operation were using edge $u_Nu$, $uu_S$, or both. If only $u_N$ is handed over to $u^{\prime}_E$ then any cycle using $u_Nuu_W$ is replaced by a cycle using $u_Nu^{\prime}_Euu_W$, any using $u_Nuu_E$ by one using $u_Nu^{\prime}_Eu_E$, and any using $u_Nuu_S$ by one using $u_Nu^{\prime}_Euu_S$. The case in which only $u_S$ is handed over is symmetric. If both $u_N$ and $u_S$ are handed over to $u^{\prime}_E$ then the only difference is that any cycle using $u_Nuu_S$ is now replaced by one using $u_Nu^{\prime}_Eu_S$. It follows that there is a one-to-one correspondence between previous and new cycles due to neighbor handover, which gives the required invariant.
\end{proof}

\subsection{Problem Definition}\label{subsec:problem-def}

In this paper, we study a reachability problem between classes of shapes through growth. The definition of the problem is the same for all growth models of Definition~\ref{def:models}.

Let $\mathcal{I}$ be a class of initial shapes and $\mathcal{F}$ a class of final shapes. We want to determine a bound $\tau$ such that for all $S_0\in \mathcal{I}$ and all $S\in \mathcal{F}$ there is a growth process $\sigma$ that grows $S$ from $S_0$ in $\tau$ time steps.

Given that our focus is on exponential growth, upper bounds must be of the form $\tau=O(\log n)$ or $\tau=(poly)\log n$. As there is a straightforward $\Omega(\log n)$ lower bound, non-trivial lower bounds should be at least $\omega(\log n)$. In all instances of the problem studied in this paper, at least one of $\mathcal{I}$ and $\mathcal{F}$ is a singleton, the initial shape typically being a single node. Our upper bounds are constructive: for each instance of the problem an algorithm is presented, which for every $(S_0,S)$ from the respective classes gives a process that grows $S$ from $S_0$ in $\tau$ time steps.

\subsection{Other Basic Notions and Properties}\label{subsec:basic-def}

Some of our results concern shapes drawn from special graph classes, such as paths, spirals, and staircases, which we now define alongside other basic notions and properties.

\emph{Turning point}. 
A node $u$ of a shape $S$ is called a \emph{turning point} 
if either $u$ is a leaf or there are a two neighbors $v_1$ and $v_2$ of $u$ such that $v_1u$ is perpendicular to $uv_2$. In the special case of a path shape $P=( u_1, u_2, \ldots, u_n )$, a node $u_i$ is a turning point if either $i \in \{1,n\}$ or $u_{i-1}u_i$ is perpendicular to $u_iu_{i+1}$. 
For uniformity of our arguments, we add the endpoints of a path to the set of its turning points. The nodes between any two consecutive turning points of a path form a line segment. A direction of an internal turning point $d(u_i)$ where $1\leq i < n$ is \emph{left} if the orientation changes from $d(u_i)$ to $d(u_{i+1})$ in a \emph{counterclockwise} or \emph{right} if the orientation changes from $d(u_{i})$ to $d(u_{i+1})$ in a \emph{clockwise} manner. 
Note that when a path $P$ is a sub-shape of a shape $S$, there might be turning points of $S$ in $P$ which are not turning points of $P$. We use the terms \emph{turning points} and \emph{turning points of} $P$ to distinguish between the two when ambiguity can arise. 

A \emph{column} or \emph{row of a shape} $S$ is a column or row 
of the grid occupied by nodes of $S$.
\emph{Compressible columns/rows}. We call a column or row of a shape $S$ \emph{compressible} if it contains no turning points and \emph{incompressible} otherwise. A compressible column contains no vertical line segments and intersects one or more horizontal line segments. 
The former holds because a vertical line segment must either end or turn within its column and the latter because, otherwise, the column would not contain any nodes of $S$. If $(C_l, C_{l+1},\ldots, C_{r})$ is a sequence of $r-l+1$ consecutive compressible columns of $S$, then any horizontal line segments $s$ intersected by one of these columns must be intersected by all. It follows that any such segment $s$ has length $r-l+1$ in $[C_l,C_r]$ and has a left neighbor $u_{l-1}(s)$ in $C_{l-1}$ and a right neighbor $u_{r+1}(s)$ in $C_{r+1}$. Similarly, for rows.

\emph{Compression}. 
Let $(C_l,C_{l+1},\ldots,C_{r})$ be a maximal sequence of consecutive compressible columns of $S$. A \emph{compression operation} on the sequence deletes the columns $C_l,C_{l+1},\ldots,C_{r}$ and for every horizontal segment $s$ intersected by these columns it connects $u_{l-1}(s)$ to $u_{r+1}(s)$. Compression on rows is defined analogously, by deleting rows in order to compress along vertical segments. 

\emph{Incompressible shape}. 
A shape is called \emph{incompressible} if it has no compressible columns or rows. We denote by $i(S)$ the incompressible shape obtained after compressing all compressible columns and rows of a shape $S$.

\begin{proposition}\label{proposition:compressed-shape}
Let $S$ be any shape having $k$ turning points and $l$ line segments. Then $i(S)$ satisfies the following properties:  
\begin{itemize}
\item $i(S)$ has $k$ turning points and $l$ line segments.
\item There is a one-to-one correspondence between the turning points and the line segments of $i(S)$ and $S$, such that $i(S)$ and $S$ are geometrically equivalent up to differences in the lengths of their line segments. In particular, the length of a line segment
in $i(S)$ is upper bounded by the length of its corresponding line segment 
in $S$. Moreover, the relative order of turning points with respect to each other is the same in both $S$ and $i(S)$.  
\item $i(S)$ fits within a $k\times k$ square area.
\end{itemize}
\end{proposition}
\begin{proof}
The first two properties hold because the only operation used to obtain $i(S)$ from $S$ is compression, which reduces the length of some of the line segments  
of $S$ and maintains the line segment endpoints. The third property follows by observing that every column and row of $i(S)$ must contain at least one turning point and there is a total of $k$ turning points.
\end{proof}

\begin{definition}[Staircase]\label{def:staircase}
A \emph{staircase} is a path whose turning points, when ordered from one endpoint to the other, alternate between two clockwise- or counterclockwise- consecutive cardinal directions.
\end{definition}
\begin{definition}[Spiral]\label{def:spiral}
    A \emph{spiral} is a path whose turning points, when ordered from one endpoint to the other, follow a continuous and unidirectional sequence of consecutive cardinal directions in either a clockwise or counterclockwise manner.  
\end{definition}

A spiral can be \emph{single} or \emph{double}. In what follows, we restrict attention to single spirals which always have an \emph{internal} and an \emph{external} endpoint. Let $(tp_1,tp_2,\ldots,tp_k)$ be the order of the $k$ turning points of a spiral, $tp_1$ being its internal and $tp_k$ its external endpoint, respectively. The consecutive turning points $tp_{i}$, $tp_{i+1}$, for $i\in \{1,2,\ldots,k-1\}$, define the $i$th line segment 
$s_i$ of the spiral. We denote by $l(s_i)$ the length of segment $s_i$ in edges.

\begin{proposition}\label{pro:spiral-inequalities}
A spiral satisfies $l(s_i)>l(s_{i-2})$, for all $i\in \{3,4,\ldots,k-1\}$. An \emph{incompressible spiral} satisfies $l(s_1)=l(s_2)=1$ and $l(s_i)=l(s_{i-2})+1$, for all $i\in \{3,4,\ldots,k-1\}$.\footnote{Note that there is an exception to this definition of incompressible spirals: the spirals in which the compression of a single row or column is blocked by the presence of the external endpoint in it. We disregard this case for clarity and it is straightforward to extend our results to include it.}
\end{proposition}
\begin{proof}
The first statement follows directly from the definition of a spiral restricted to single spirals. The second statement is true because $l(s_i)>l(s_{i-2})+1$ and the rows between $tp_{i-2}$ and $tp_{i+1}$ being incompressible can only hold at the same time if at least one of $tp_{i+3}, tp_{i+5}, \cdots, tp_{i+2j+1}$ is contained within that row, for $j \geq 1$. Let $tp_{i+2j+1}$ be such a vertex with the smallest index. However, these would imply that $l(s_{i+2j})<l(s_{i+2j-2})$, thus contradicting the definition of a spiral. 
\end{proof}

We denote by $\mathcal{C}_k$ the class of all shapes having a rooted spanning tree $T$ such that every root-to-leaf path in $T$ has at most $k$ turns. A \emph{fast line growth process} begins from a single node and by successively doubling all nodes grows a line segment of length $n$ in $O(\log n)$ time steps. 
Fast line growth is used as a sub-process in most of our constructions in order to efficiently grow line segments of a shape.

\section{Connectivity Graph Model}\label{sec:connectivity-graph-model}
According to Proposition{~\ref{prop:models-diff}} and starting growth from a single node, the class of shapes that can be grown in this model is limited to tree structures only. In the absence of the neighbor handover property, there is no difference between cycle-breaking and cycle-preserving growth, as there are no cycles. Thus, in Section{~\ref{sec:adj-graph-model}}, we study the capabilities of these growth processes within the adjacency model, and we start this section by focusing on efficiently growing trees.

\subsection{Upper Bounds}\label{subsec:upper-bound}

\begin{lemma}\label{lemma:compress-shape}
In the connectivity model, any shape $S$ can be grown from $i(S)$ in $O(\log n)$ time steps, $n$ being the number of nodes in $S$.
\end{lemma}

\begin{proof}
By Proposition \ref{proposition:compressed-shape}, $S$ has the same turning points as $i(S)$, their only difference being the lengths of line segments connecting two turning points. To obtain $S$ from $i(S)$ it is sufficient to grow each of these segments to reach its final length in $S$. By doing so in parallel and through a fast line growth process for each segment, all segments can be grown in a number of time steps which is logarithmic in the maximum length to be grown, that is, $O(\log n)$ time steps in the worst case. Recall that as the relative order of turning points with respect to each other is the same in both $S$ and $i(S)$, so the above procedure does not result in any collision. 
\end{proof}

\begin{corollary}\label{coro:growing-compressed-shape}
In the connectivity model, any shape $S$ on $n$ nodes can be grown from a single node in $\tau_{min}(i(S))+O(\log n)$ time steps, where $\tau_{min}(i(S))$ is the minimum number of time steps in which $i(S)$ can be grown from a single node.   
\end{corollary}

Let $T$ be any tree shape having $O(k)$ turning points on every root-to-leaf path. We can combine BFS on line segments with fast line growth to grow $T$ in $O(k\log n)$ time steps. We now describe this approach. 
Every root-to-leaf path of $T$ is an alternating sequence of turning points and line segments of the form $(u_0,s_1,u_1,s_2,u_2,\ldots,u_{j-1},s_{j},u_{j})$, where $j\leq k-1$. We refer to $u_{i-1}$ and $u_i$ as the \emph{endpoints} of segment $s_i$ closer to, and further from $u_0$, respectively, even though we do not formally include them in $s_i$. 
The \emph{segment distance} of line segment $s_i$ (equivalently of its furthest endpoint $u_i$) from $u_0$ is equal to $i$ and corresponds to the number of line segments between $u_0$ and $u_i$.\\ 

\noindent\textbf{BFS on line segments}\\
\noindent Input: tree $T$ rooted at $u_0$\\
\noindent Output: growth process for $T$ starting from a single node at $u_0$\\
\noindent Repeat in phases $i=1,2,\ldots$ until $T$ is constructed: 
\begin{itemize}
\item Do in parallel for every line segment $s_i$ at segment distance $i$ from $u_0$ in $T$, where $u_{i-1}$, $u_i$ are its endpoints:
     \begin{itemize}
     \item Starting from $u_{i-1}$, grow $(s_i,u_i)$ by fast line growth. 
    \end{itemize}
\end{itemize}

\begin{theorem}\label{theo:shape-spanning-tree}
Let $T$ be any tree having $O(k)$ turns on every root-to-leaf path. \emph{BFS on line segments} can grow $T$ from a single node in $O(k \log n)$ time steps in the connectivity graph model.
\end{theorem}
\begin{proof}
Given $T$ rooted at $u_0$, the process returned by the algorithm starts from a single node at $u_0$. Then in each phase $i\geq 1$ it grows in parallel all line segments of $T$ at segment distance $i$ from $u_0$. By induction on $i$, it can be shown that this is a collision-free and correct construction of $T$. This is because, assuming that $u_{i-1}$ has been positioned correctly in phase $i-1$, the algorithm grows $(s_i,u_i)$ on a line of empty points corresponding to their positions in $T$. Moreover, the connectivity graph model ensures that the growth of a segment does not interfere with segments adjacent to it. The bound on the time steps of the returned process holds because the number of phases is upper bounded by $O(k)$ and every phase grows in parallel one or more segments. The length of segments is upper bounded by $n$, thus, any segment can be grown by fast line growth in $O(\log n)$ time steps.
\end{proof}

The above bound can be rather crude in some cases; for example, on a path consisting of one line segment of length $\Theta(n)$ and $O(\log n)$ segments of constant length each, the BFS on line segments yields an actual $O(\log n)$ time steps compared to the $O(\log^2 n)$ of the above analysis. However, this cannot be improved by a better analysis in the worst case, even for paths.

\begin{theorem}\label{the:BFS-counterexample}
Let $P$ be any path with $\Theta(\log n)$ turning points, such that each of the line segments of $P$ has length $\Theta(n/\log n)$. The process returned by BFS on line segments for $P$ uses $\Theta(\log^2 n)$ time steps.
\end{theorem}
\begin{proof}
The algorithm uses $\Theta(\log n)$ phases to grow the segments and on each segment $s_i$ it takes $\Theta(\log l(s_i))=\Theta(\log(n/\log n))$ time steps. So, the total number of time steps is $\Theta(\log n\log(n/\log n))=\Theta(\log n(\log n-\log\log n))=\Theta(\log^2 n)$. 
\end{proof}

The following approach is designed to efficiently grow shapes with uniform growth patterns, such as spirals.\\

\noindent\textbf{Pipelined BFS}\\
\noindent Input: any shape $S$ with a uniform growth pattern\\
\noindent Output: growth process $\sigma$ for $S$\\
\noindent It consists of two iterative phases and proceeds regularly until no new turning points can be added and all line segments reach their final length.
\begin{itemize}
    \item \textbf{Construction and Waiting Phase}: during this phase, it builds at most three turning points $tp_i$ (where $i$ ranges from 1 to 3) in an order that follows the geometry of a shape $S$, each turning point $tp_i$ and its subsequent $tp_{i+1}$ defines the $i^{th}$ line segment $s_i$ of the shape $S$.
    \item \textbf{Parallel Growing Phase}: in this phase, the partially formed structure from 
    the previous step grows in parallel, except for those line segments that have reached their full length.
\end{itemize}
\begin{theorem}\label{theo:spiral-logn-turns}
    Let $P$ be a spiral with $k=O(\log n)$ turning points and $( tp_1, tp_2, $ $\dots, tp_k )$ the order of turning points of $P$, with $tp_k$ being the external endpoint. Starting from $tp_k$, Pipelined BFS can grow $P$ in $O(\log n)$ time steps in the connectivity graph model.
\end{theorem}
\begin{proof}
The algorithm grows $P$, starting from the external endpoint $tp_k$ and extending toward the internal endpoint $tp_1$. In phase $i=1$, it grows the first three turning points $tp_k, tp_{k-1}$ and $tp_{k-2}$, forming segments $s_{k-1}$ and $s_{k-2}$, as the outermost layer of $P$. This phase efficiently grows the segments in parallel and simultaneously creates the space for the subsequent inner layers. The parallelism of the \emph{Pipelined-BFS} ensures that subsequent layers are initiated while the growth of one layer is completed. Due to the uniform geometry of the spiral, this approach consistently grows the line segments in a pipelined fashion. As the spiral path $P$ has $\log n$ layers, we consume $\log n $ waiting time for each layer to grow. Therefore, the total number of time steps required to grow $P$ is bounded by the number of turning points $\log n$ plus the waiting time for each layer to grow, which is $\log n$; thus, the spiral path $P$ can be grown in $O(\log n)$ steps.
\end{proof}

\subsection{Lower Bounds}\label{subsec:lower-bound-connectivity-graph-model}

There is a generic lower bound which follows immediately due to the limit on the rate at which nodes can generate new nodes; a limit inherent to all the models we consider.

\begin{proposition}\label{pro:general-lower-bound-2d}
Let $S$ be any shape on $n$ nodes. In any of the considered models, any growth process for $S$ starting from a single node, requires $\Omega(\log n)$ time steps.
\end{proposition}

\begin{proof}
The lower bound follows by observing that in every time step, each node can generate a number of nodes which is upper bounded by its maximum degree in the model, which is 4 on the two-dimensional square grid. In the special case that we typically consider, in which every node is allowed to generate at most one node per time step, this upper bound becomes 1. In all cases, if $N_i$ is the number of nodes in the beginning of time step $i\geq 1$, the number of nodes at the end of time step $i$ is at most $cN_{i}$, for a constant $c\leq 5$, implying inductively that the total number of nodes at the end of time step $i$ can be at most $c^i\leq 5^i$. It follows that $\Omega(\log n)$ time steps are required to generate $n$ nodes when starting from a single node.
\end{proof}

Refined lower bounds can be obtained as a function of a parameter of the final shape or by restricting attention to special classes of final shapes. In Theorem \ref{the:lower-bound-trees}, we prove that the number of turning points of root-to-leaf paths of trees gives a lower bound on the number of time steps of any process for them. By focusing on path shapes, in Theorem \ref{the:lowerBound} we obtain an $\Omega(k\log k)$ lower bound on the number of time steps of any process for a specific worst-case type of path, $k$ being the total number of turning points of the path.

\subsubsection{Trees}\label{subsubsec:lb-trees}

Let $S=(V,E)$ be a shape. Any growth process $\sigma$ for $S$ induces a relation $\rightarrow_{\sigma}$ on $V$, where $u\stackrel{t}\rightarrow_{\sigma} v$ iff node $u$ \emph{generates} node $v$ at time step $t$. We also write $u\rightarrow_{\sigma} v$ to mean $u\stackrel{t}\rightarrow_{\sigma} v$ for some $t\geq 1$ and $u\rightsquigarrow_{\sigma} v$ iff $u=u_1\rightarrow_{\sigma}u_2\rightarrow_{\sigma}\cdots\rightarrow_{\sigma}u_l=v$, for some $l\geq 2$. We omit $\sigma$, writing just $u\stackrel{t}\rightarrow v$, $u\rightarrow v$, or $u\rightsquigarrow v$ when the growth process is clear from context or when referring to any growth process. The relation $\rightarrow_{\sigma}$ defines a graph $G_{\rightarrow_{\sigma}}=(V,E_{\rightarrow_{\sigma}})$. The following proposition gives some properties of this graph.

\begin{proposition}\label{prop:tree-property}
$G_{\rightarrow_{\sigma}}$ is an out-tree spanning $V$, rooted at the initial node $u_0$.
\end{proposition}

\begin{proof}
Observe that every node $u\in V$ is generated by $\sigma$ at some time step $t_l$. Therefore, by induction on time steps, there must be a sequence $u_0\stackrel{t_1}\rightarrow u_1\stackrel{t_2}\rightarrow\cdots u_{l-1}\stackrel{t_l}\rightarrow u$ starting at the initial node $u_0$. Thus, there exists a path from $u_0$ to $u$ in $G_{\rightarrow_{\sigma}}$. It follows that $G_{\rightarrow_{\sigma}}$ is connected, spanning $V$, and is rooted at $u_0$.

It remains to show that $G_{\rightarrow_{\sigma}}$ is acyclic. For any $u_1\in V$, a cycle $u_1,u_2,\ldots,u_l,u_1$, $l\geq 2$, would imply that there is a sequence $u_1\stackrel{t_1}\rightarrow u_2\stackrel{t_2}\rightarrow\cdots u_l\stackrel{t_l}\rightarrow u_1$. Note that, in general, $u\stackrel{t}\rightarrow v\stackrel{t'}\rightarrow w$ implies $t'>t$ as, by assumption of the model, a node $v$ can only generate a node after $v$ is itself generated. Therefore, it must hold that $t_1<t_2<\cdots <t_l$ and the existence of the cycle $u_1,u_2,\ldots,u_l,u_1$ would imply the contradictory statement that node $u_1$ was generated both before time step $t_1$ and after $t_l>t_1$.
\end{proof}

We now define a relation $\mapsto_{\sigma}$ induced by $\rightarrow_{\sigma}$ on the turning points of tree shapes. In particular, given a tree $T$ and a growth process $\sigma$ for $T$, for any two turning points $u,v$ of $T$ we write $u\stackrel{t}\mapsto_{\sigma} v$ iff (i) $u\stackrel{t}\rightarrow_{\sigma} v$ or (ii) $u\rightsquigarrow u'\stackrel{t}\rightarrow_{\sigma} v$ and $u$, $u'$, $v$ are on the same line segment at the end of time step $t$.
We again write $u\mapsto_{\sigma} v$ to mean $u\stackrel{t}\mapsto_{\sigma} v$ for some $t\geq 1$, and will often omit $\sigma$. The relation $\mapsto_{\sigma}$ defines a graph $G_{\mapsto_{\sigma}}$ on turning points whose structure has a quite strong dependence on the structure of $T$.

\begin{lemma}\label{lemma:graph-equality}
In the connectivity model, let $T=(V,E)$ be a tree and $\sigma$ a growth process for $T$ starting from $u_0\in V$. For any root-to-leaf path $( u_0,u_1,\ldots,u_l)$ of $T$, where the $u_i$s are restricted to the turning points of the path, $u_0\mapsto u_1\mapsto\cdots \mapsto u_l$ holds. 
\end{lemma}
\begin{proof}

We prove the statement by induction on the index $i$ of the turning points on the path. The base case, for $i=0$, holds trivially as node $u_0$ of $T$ is by assumption the initial node of the process. Assume that the statement holds for some $i\geq 0$, that is, for the $( u_0,u_1,\ldots,u_i)$ prefix of the $( u_0,u_1,\ldots,u_i,u_{i+1},\ldots,u_l)$ root-to-leaf path of $T$, the process $\sigma$ satisfies $u_0\mapsto u_1\mapsto\cdots \mapsto u_i$. We will show that this must imply that $u_i\mapsto u_{i+1}$. For the sake of contradiction, assume not. Then, due to connectivity of $G_{\mapsto_{\sigma}}$ it must hold that $u\mapsto u_{i+1}$ for some turning point $u\neq u_i$ of $T$. There are two possible cases:
\begin{enumerate}
\item $u\in\{u_0,\ldots,u_{i-1}\}$: As $u\stackrel{t}\mapsto u_{i+1}$ for some $t$, it holds that $u$ and $u_{i+1}$ are joined by a line segment at the end of time step $t$. In the connectivity model, a line segment, once formed, is forever maintained and can only grow, implying that $u$ and $u_{i+1}$ must be connected by a line segment in $T$. But this can only hold if $u=u_{i}$, without contradicting one or more $u_j$s being turning points of the path or their order on the path. The latter, in turn, contradicts $u$'s assumed inclusion in $\{u_0,\ldots,u_{i-1}\}$.
\item $u\notin\{u_0,\ldots,u_i\}$: In this case, it holds that $u_0\rightsquigarrow u\mapsto u_{i+1}$ and also by the inductive hypothesis that $u_0\mapsto u_1\mapsto\cdots \mapsto u_{i-1}\mapsto u_i$. Then, $( u_0,\ldots,u,u_{i+1})$ and $( u_0,u_1,\ldots,u_{i-1},u_i,u_{i+1})$ must both be paths of $T$, because for any two turning points $v,v'$, $v\mapsto v'$ implies a line segment between $v$ and $v'$ in $T$. It follows that $u_0,u_1,\ldots,u_i, $ $u_{i+1},u,\ldots,u_0$ (the line segments between turning points inclusive) forms at least one cycle in $T$, contradicting the fact that $T$ must be a tree in the connectivity model. 
\end{enumerate}
We conclude that the inductive step $u_i\mapsto u_{i+1}$ must hold and this completes the proof.
\end{proof}

Lemma \ref{lemma:graph-equality} says that the turning points of any tree shape $T$ that is constructed in the connectivity model, can only be generated in their root-to-leaf order in a rooted version of $T$. Note that this does not necessarily hold for non-turning points of $T$. The construction of the latter can bypass their actual order in $T$ by, for example, growing line segments of such points through a fast line growth process. 

\begin{theorem}\label{the:lower-bound-trees}
Let $T=(V,E)$ be a tree and $k$ any positive integer satisfying that for every root $u_0\in V$ there is a root-to-leaf path in $T$ containing at least $k$ turning points. Then any growth process $\sigma$ for $T$ in the connectivity model requires at least $k-1$ time steps. This lower bound is maximized for the maximum such $k$.
\end{theorem}
\begin{proof}
The process $\sigma$ must start from a node $u_0$. Take any mapping of $u_0$ to a node of $T$ and call the latter $u_0$ too. $T$ rooted at $u_0$ has a root-to-leaf path $( u_0,u_1,\ldots,u_l)$ restricted on turning points, where $l\geq k$. Due to Lemma \ref{lemma:graph-equality}, $\sigma$ must satisfy $u_0\stackrel{t_1}\mapsto u_1\stackrel{t_2}\mapsto\cdots \stackrel{t_{l-1}}\mapsto u_l$. The latter implies $1\leq t_1<t_2<\cdots<t_{l-1}$, thus, $t_{l-1}\geq l-1\geq k-1$.
\end{proof}

\subsubsection{Paths}\label{subsubsec:lb-paths}

Let $P$ be a path with $k$ turning points. Let $\sigma$ be a process that grows $P$ from a single node. Without loss of generality, we can assume that $\sigma$  starts from a turning point of the path $P$. We now give a few observations and lemmas concerning some properties of $\sigma$. Recall that an edge, once generated, cannot be deleted in this model. This immediately implies the following observation.

\begin{observation}\label{obs:growAlong}
    A node can grow in at most its degree many different directions. Moreover, once a node has its degree many neighbors in the path constructed by $\sigma$, it can only grow along one of its incident edges in the path.
\end{observation}

As there exists a unique subpath between any two vertices in a path, this fact, together with the above observation gives the following observation.

\begin{observation}\label{obs:noTurn}
    Let $x$ and $z$ be any two vertices of $P$ such that there exists a line segment between them in the path constructed so far by $\sigma$. Then, all the vertices on the subpath between $x$ and $z$ in $P$ will lie on a line segment in the final path constructed by $\sigma$. 
\end{observation}

We now give the following lemma concerning the order in which the turning points of $P$ are generated by $\sigma$. 
\begin{lemma}\label{lem:inOrderTP}
    Let $P$ be a path between $u$ and $v$ with $k$ turning points. Let $(tp_1, tp_2, $ $\dots, tp_k)$ be the order of turning points of $P$ from $u$ to $v$. Let $\sigma$ be any process that grows $P$ from a single node starting from the turning point $tp_i$. Then, the sets $\{tp_{i+1}, tp_{i+2}, \dots, tp_k\}$ and $\{tp_1, tp_2, \dots, tp_{i-1}\}$ of turning points are generated in the order $(tp_{i+1}, tp_{i+2}, \dots, tp_k )$ and $(tp_{i-1}, tp_{i-2}, \dots, tp_1 )$, respectively by $\sigma$. Moreover, $\sigma$ respects the direction of $P$ at every node while generating the next node from it.
\end{lemma}

\begin{proof}
Recall that, an edge, once generated, cannot be deleted in this model. This in turn means that a node can grow in at most its degree many different directions. Moreover, once a node has degree many neighbors, it can only grow along an incident edge.

We first prove that $\sigma$ respects the direction of $P$ at every node while generating the next node from it. As the direction makes sense only when the node already has a neighbor, we prove the statement for the nodes that grow at time step $2$ or later. Let $\sigma$ grow the node $v$ at time step $t \geq 2$. Assume for contradiction, $\sigma$ does not respect the direction of $P$ at $v$ while generating the next node $u$. Then, once $u$ is generated, the degree of $v$ is $2$ in the path constructed by $\sigma$ so far. By Observation~\ref{obs:growAlong}, we get that $\sigma$ can never generate a neighbor of $v$ in the desired direction, a contradiction. Thus, $\sigma$ always respects the direction of $P$ at every node while generating the next node from it.

We now prove the property regarding the order of generation of turning points. We prove that the set $\{tp_{i+1}, tp_{i+2}, \dots, tp_k\}$ is generated in the order $(tp_{i+1}, tp_{i+2}, \dots, tp_k)$. The proof for the set $\{tp_1, tp_2, \dots, tp_{i-1}\}$ is similar. Let $tp_{j+1}$ be the first turning point that was not generated in the desired order, for $j \geq i$. Moreover, let $tp_k$ be the turning point that was generated after $tp_j$, for $k > j$. This implies that there exists a subpath $P'$ from $tp_j$ to $tp_k$ of the path constructed so far by $\sigma$ which does not contain any other turning points, i.e., $P'$ is drawn as a line segment. As $tp_{j+1}$ lies between $tp_j$ and $tp_k$ in $P$, by Observation~\ref{obs:noTurn}, we get that $\sigma$ can never generate the two neighbors of $tp_{j+1}$ in different directions. This contradicts the fact that $tp_{j+1}$ is a turning point of $P$. Thus, we conclude that the set $\{tp_{i+1}, tp_{i+2}, \dots, tp_k\}$ is generated in the order $(tp_{i+1}, tp_{i+2}, \dots, tp_k )$.
\end{proof}

Let $P$ be an incompressible spiral path between $u$ and $v$ with $k$ turning points. 
Moreover, let $u$ be the internal endpoint of $P$. We now give the following lemma about the lower bound on the number of time steps taken by any process that grows $P$ from a single node starting from $u$.

\begin{figure}[t]
     \centering
     \begin{subfigure}[t]{.33\textwidth}
         \centering
        \includegraphics[page=1]{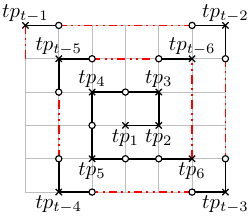}
         \caption{}
        \label{fig:incompressiblePathA}
    \end{subfigure}\hfil
     \begin{subfigure}[t]{.33\textwidth}
        \centering
         \includegraphics[page=2]{Figures/geoGrowth}
        \caption{}
         \label{fig:incompressiblePathB}
     \end{subfigure}\hfil
     \begin{subfigure}[t]{.33\textwidth}
         \centering
          \includegraphics[page=3]{Figures/geoGrowth}
          \caption{}
         \label{fig:incompressiblePathC}
     \end{subfigure}
    
     \caption{(a) An illustration of turning points (drawn by cross points) of an incompressible spiral path $P$. (b) An illustration of the path $\widehat{P}$ constructed by $\sigma$ in the Base Case in the proof of Lemma~\ref{lem:lowerBound2}. The black (blue) part of $\widehat{P}$ shows the subpath, which is the same as (different from) $P$. The dotted green line shows the direction in which $tp_5$ will be generated. (c) An illustration of the path $\widehat{P}$ constructed by $\sigma$ in the Inductive step in the proof of Lemma~\ref{lem:lowerBound2} for $t > 6$. The black (blue) part of $\widehat{P}$ shows the subpath, which is the same as (different from) $P$. The dotted green line shows the direction in which $tp_t$ will be generated.}
     \label{fig:incompressiblePath}
 \end{figure}

\begin{lemma}\label{lem:lowerBound1}
    Let $P$ be an incompressible spiral path between $u$ and $v$ with $k$ turning points. Moreover, let $u$ be the internal endpoint of $P$. Let $\sigma$ be any process that grows $P$ from a single node starting from $u$. Then, $\sigma$ requires $\Omega(k\log k)$ time steps.
\end{lemma}

Let $(tp_1 = u, tp_2, $ $\dots, tp_k = v )$ be the order of turning points of $P$ from $u$ to $v$. By Lemma~\ref{lem:inOrderTP}, we know that $\sigma$ generates the turning points in the order $( tp_1 = u, tp_2, $ $\dots, tp_k = v )$. Let $GT_j$ be the time step when the turning point $tp_j$ was generated by $\sigma$, for any $j \geq 2$. Let $\widehat{P}(t)$ be the path constructed by $\sigma$ after time step $t$. Further, let $a$ and $b$ be two vertices of $P$. We denote by $P[a,b]$ the path between $a$ and $b$ (including both $a$ and $b$) of $P$. Moreover, we denote by $|a-b|_P$ the number of edges in $P[a,b]$. Also, we denote by $X(a,P)$ the x-coordinate of the vertex $a$ in $P$. To prove the above lemma, we first prove the following lemma about the path constructed by $\sigma$.

\begin{lemma}\label{lem:lowerBound2}
    For any $j \geq 5$, the path $\widehat{P}(GT_j-1)$ grown by $\sigma$ till time step $GT_j - 1$ should be the same as the subpath $P[tp_1, tp_{j-1}]$ of $P$ between $tp_1=u$ and $tp_{j-1}$. 
\end{lemma}

\begin{proof}
     We prove the statement by induction on $j$.\\
    \noindent{\bf Base case $(j=5)$.} Recall that, as $P$ is incompressible, $|tp_2 - tp_1|_P = |tp_3 - tp_2|_P = 1$. Thus $GT_3 = 2$ and $\widehat{P}(GT_3) = P[tp_1, tp_3]$. Assume for contradiction that the lemma is not true for $j=5$. This means that $\widehat{P}(GT_5 - 1)$ is a subpath of $P[tp_1,tp_4]$. By Lemma~\ref{lem:inOrderTP}, we know that $GT_5 > GT_4 > GT_3$. As $\widehat{P}(GT_3) = P[tp_1, tp_3]$, we get that $P[tp_1, tp_3]$ is a subpath of $\widehat{P}(GT_5 - 1)$. Combining this fact with the fact that $\widehat{P}(GT_5 - 1)$ is a subpath of $P[tp_1,tp_4]$, we get that $1 \leq |tp_4 - tp_3|_{\widehat{P}(GT_5 - 1)} < |tp_4 - tp_3|_P = 2$. This implies that $|tp_4 - tp_3|_{\widehat{P}(GT_5-1)} = 1$. This further mean that $X(tp_4, \widehat{P}(GT_5 - 1)) = X(tp_1, \widehat{P}(GT_5 - 1))$. By Lemma~\ref{lem:inOrderTP}, we get that $\sigma$ respects the direction of $P$ at every node. Therefore, when $tp_5$ is generated it will collide with $tp_1$, a contradiction (e.g., see Figure~\ref{fig:incompressiblePathB}). So, the lemma is true for $j = 5$.

    \noindent{\bf Inductive hypothesis.} Suppose that the lemma is true for $j = t - 1 \geq 5$.

    \noindent{\bf Inductive step.} We need to prove that the lemma is true for $j = t \geq 6$. Assume for contradiction that the lemma is not true for $j$. This means that $\widehat{P}(GT_t - 1)$ is a subpath of $P[tp_1,tp_{t-1}]$. By Lemma~\ref{lem:inOrderTP}, we know that $GT_t > GT_{t-1} > GT_{t-2}$. By the inductive hypothesis, we know that $\widehat{P}(GT_{t-1} - 1) = P[tp_1, tp_{t-2}]$. This implies that $P[tp_1, tp_{t-2}]$ is a subpath of $\widehat{P}(GT_t - 1)$. Combining this fact and the fact that $\widehat{P}(GT_t - 1)$ is a subpath of $P[tp_1,tp_{t-1}]$, we get that $1 \leq |tp_{t-1} - tp_{t-2}|_{\widehat{P}(GT_t - 1)} < |tp_{t-1} - tp_{t-2}|_P = \lfloor \frac{t-1}{2} \rfloor$. This further implies that either $t=6$ and $X(tp_5, \widehat{P}(GT_6 - 1)) = X(tp_1, \widehat{P}(GT_6 - 1))$, or $X(tp_{t-5}, \widehat{P}(GT_t - 1)) \leq X(tp_{t-1}, \widehat{P}(GT_t - 1)) < X(tp_{t-6}, \widehat{P}(GT_t - 1))$. By Lemma~\ref{lem:inOrderTP}, we get that $\sigma$ respects the direction of $P$ at every node. Therefore, when $tp_t$ is generated, it will collide with a node on the subpath of $\widehat{P}(GT_{t-1} - 1)$ between $tp_{t-5}$ and $tp_{t-6}$, a contradiction (e.g., see Figure~\ref{fig:incompressiblePathC}). So, the lemma is true for $j=t$.
 \end{proof}

We now give the proof of Lemma~\ref{lem:lowerBound1} using Lemma~\ref{lem:lowerBound2}. 

\begin{proof}[Proof of Lemma~\ref{lem:lowerBound1}]
Let $ST_j$ be the time taken by $\sigma$ to grow the path $P[tp_1, tp_j]$ starting from $tp_1$, for any $j \geq 2$. Then, by Lemma~\ref{lem:lowerBound2}, we get that $GT_j \geq ST_{j-1} + 1$, for any $j \geq 5$. Moreover, by Lemma~\ref{lem:lowerBound2}, we know when $tp_j$ is generated, the subpath from $tp_1$ to $tp_{j-1}$ is already generated by $\sigma$. So, the difference between $\widehat{P}(GT_j)$ and $P[tp_1, tp_j]$ is the length of the subpath between $tp_{j-1}$ and $tp_j$ in both the paths. As we know the subpath between $tp_{j-1}$ and $tp_j$ is a line segment in $P$, we can grow it in $\log(|tp_j - tp_{j-1}|_P)$ time steps. This implies that, $ST_j = GT_j + \log(|tp_j - tp_{j-1}|_P)$. Combining the two equations, we get that $ST_j \geq ST_{j-1} + 1 + \log(|tp_j - tp_{j-1}|_P)$. It is easy to observe that $ST_4 = 4$. Thus, by solving the recursive relation, we get that $ST_j = \Omega(k\log k)$. This proves the lemma.
\end{proof}

We now give the main theorem of this section.
\begin{figure}
    \centering
    \includegraphics[page=4, width=.3\textwidth]{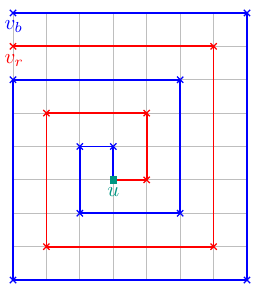}
    \caption{An illustration of an incompressible path consisting of a red and a blue spiral used in the proof of Theorem~\ref{the:lowerBound}. The green square vertex $u$ denotes the internal endpoint of both spirals.}
    \label{fig:lowerBound}
\end{figure}
\begin{theorem}\label{the:lowerBound}
    Let $\sigma$ be a process that grows a path from a single node. Then, there exists a path for which $\sigma$ takes $\Omega(k\log k)$ time steps.
\end{theorem}

\begin{proof}
    We prove the theorem by giving a path on which any process that grows it from a single node takes $\Omega(k\log k)$ time steps. We grow an incompressible path $P$ consisting of two spirals as shown in Figure~\ref{fig:lowerBound}. It is easy to observe that, due to Lemma~\ref{lem:inOrderTP}, irrespective of the starting node, $\sigma$ will grow one of the red or blue spirals from its internal endpoint $u$. Then, by a similar proof to that of Lemma~\ref{lem:lowerBound1}, we can prove that $\sigma$ takes $\Omega(k\log k)$ time steps. 
\end{proof}

\section{Adjacency Graph Model}\label{sec:adj-graph-model}

In this section, we study two types of growth: one that avoids collisions by preserving cycles and another that does so by breaking them. 
Due to cycle-preserving growth being a special case of cycle-breaking growth, positive results for the former immediately hold for the latter.

\subsection{Cycle-Preserving Growth}\label{subsec:cycle-preserving-growth}

BFS on line segments cannot be directly applied to cycle-preserving growth in the adjacency model due to the dependence between adjacent line segments. We give a modified BFS that overcomes this by growing adjacent line segments in different phases. 
For each line segment $s$ ---either horizontal or vertical--- of phase $i$ the growth process is as follows: 

\begin{itemize}
 \item Initially, we divide the line segments into sub-segments. Then, for each sub-segment $s_{i,j}$ of a line segment $s_i$, where $j=1,2,\ldots, M_i$ and $M_i$ denotes the maximum number of sub-segments of $s_i$, which is adjacent to a sub-segment $s_{i-1,j}$ grown in a previous phase, we grow $s_{i,j}$ by duplicating $s_{i-1,j}$. This process is done in parallel for all these sub-segments. The remaining sub-segments are then grown in parallel using fast line growth.
\item For any line segment $s_{i}$ that is adjacent to another line segment, also denoted as $s_{i}$, which will be grown in the same phase $i$ in parallel, we use two sub-phases: $i_{h_{even}}$ and $i_{h_{odd}}$ ($i_{v_{even}}$ and $i_{v_{odd}}$ for the vertical sub-phases). In $i_{h_{even}}$, we grow the even-row line segments, followed by the odd-row line segments in $i_{h_{odd}}$. This process is then repeated for the vertical sub-phases. 
\end{itemize}

See Figure~\ref{fig:variant-bfs} for an illustration of this process.

\begin{figure}[ht]
\centering 
\includegraphics[width=2in]{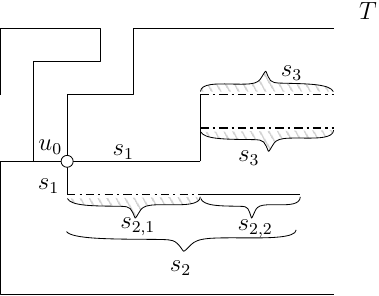}
\caption{An illustration of the modified BFS approach for growing adjacent line segments in the shape without collisions. In this example, the growth of sub-segment $s_{2,1}$ of $s_2$ during phase $i=2$ involves growing its length from the adjacent previously grown line segment $s_1$. Furthermore, when two adjacent line segments are growing in parallel, such as $s_3$, they are separated based on whether they are positioned in an even or odd row of $T$.
}
\label{fig:variant-bfs}
\end{figure}

\begin{theorem}\label{theo:tree-with-k-turns}
    For any shape $S\in \mathcal{C}_k$, $AC(S)$ can be grown in $O(k \log n)$ time steps in the adjacency graph model.  
\end{theorem}

\begin{proof}
To prove the statement, we use induction on the number of phases.
For the base case, since $s_1$ is the only line segment at this point, it covers all the paths within distance $i=1$ in $T$, so the statement holds.
For the inductive step, let us assume that after $i$ phases, the modified BFS has grown up to the $i$-th line segment of every path within distance $i$ in $T$, which corresponds to the number of turns $k$.
Then, in the $i+1$ phase, we grow the line segment 
$s_{i+1}$ of each path within distance $i+1$. For each sub-segment $s_{{i+1},j}$ of $s_{i+1}$ that is adjacent to a sub-segment $s_{i,j}$ grown in a previous phase, we can directly grow $s_{{i+1},j}$ by duplicating $s_{i,j}$ in a single time step. 

For a line segment $s_{i+1}$ that is adjacent to another line segment also scheduled to be grown in phase $i+1$, we use two sub-phases to control their growth. Let us assume without loss of generality that it is a horizontal line segment, then we have $i_{h_{even}}$ and $i_{h_{odd}}$. In $i_{h_{even}}$, we grow the even-row line segment, and in $i_{h_{odd}}$, we grow the odd-row line segment. This ensures that adjacent line segments in the same phase are grown separately without collisions.

By the induction hypothesis, after $i$ phases, we have grown up to the $i$-th line segment of every path within a distance $i$, corresponding to the number of turns $k$. Therefore, for a tree $T$ with $k$ turns, it can be grown in $O(k \log n)$ time steps. After that, the adjacency model adds edges between all adjacent nodes of $T$, which builds $AC(T)$.
Since all line segments $s_1, s_2,\ldots, s_{i}$ are constructed using the modified BFS, the structure of the tree is maintained, and no line segments collide with other line segments during the parallel growth.
 \end{proof}

As an example application of Theorem \ref{theo:tree-with-k-turns}, any compact shape whose perimeter has a bounded number of turns can be grown in $O(\log n)$ time steps in the adjacency model.

\subsection{Cycle-Breaking Growth with Neighbor Handover}\label{subsec:cycle-breaking-growth}
This growth process is characterized by its ability to break cycles within a shape to avoid collisions while maintaining global connectivity. 
Our main result in this section is a universal algorithm that efficiently provides an $O(\log n)$ time steps \emph{growth process} for any connected shape $S$. The algorithm achieves this by specifying an elimination order of the nodes within the shape and then inverting this order to produce the \emph{growth process}. It is important to note that this algorithm utilizes the \emph{neighbor handover} property, which is crucial, particularly in breaking edges and transforming neighboring nodes. 

Given that a shape $S$ has $C$ columns and $R$ rows, the \emph{elimination} algorithm consists of two sets: vertical and horizontal phases. Without loss of generality, let us assume that we start with the vertical phases:
    \begin{itemize}
        \item In the first vertical phase, count rows starting from the bottom-most row and identify odd and even rows.
        \item For every node $u$ in an odd row that has a neighbor $v$ in an even row connected by the edge $uv$, eliminate $v$ by contracting the edge $uv$ toward $u$. Then, register the eliminated or translated nodes (i.e., if there is no neighbor, a node moves down one row) in a process $\sigma$ to maintain their order.
        \item At the end of this phase, add all edges between nodes. Then proceed to the next vertical phase, recounting rows from the bottom-most row and repeating the previous steps.
    \end{itemize}
After completing the set of vertical phases, we obtain a horizontal line with a length equal to the horizontal dimension of the shape (i.e., the number of columns in $S$). Then, we perform the same steps horizontally, which results in the elimination of the horizontal line through successive halving. Finally, after completing both the vertical and horizontal phases, we reverse the order of eliminated nodes and return the growth process.

\begin{figure}[ht]
\centering 
\includegraphics[width=4in]{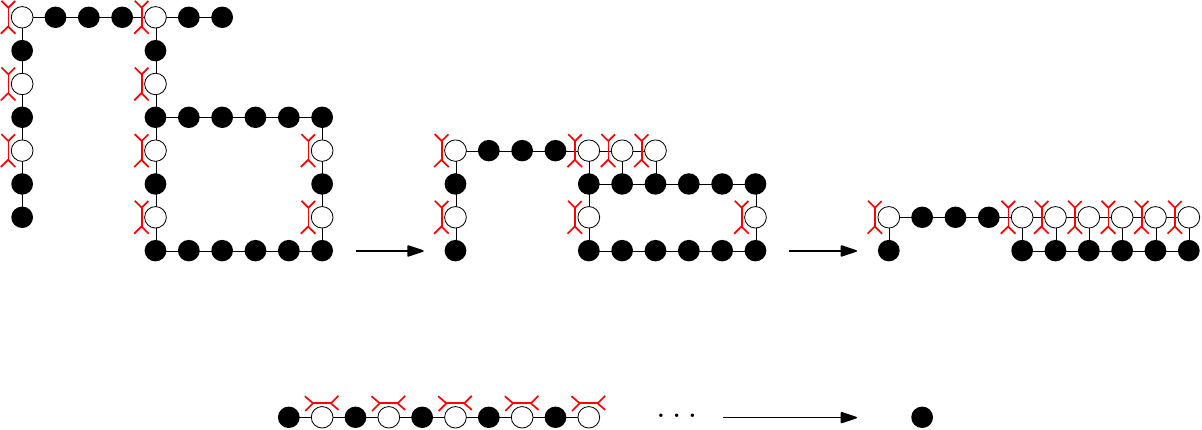}
\caption{A scenario of the Elimination algorithm for constructing a growth process $\sigma$ for a
shape $S$. The red sign indicates the elimination sequence, first vertically and then horizontally.}
\label{fig:elimination-eg}
\end{figure}

\SetAlgoNoLine
\begin{algorithm}[htbp]
\caption{Elimination algorithm }
\label{alg:elimination-algorithm}
\SetKwInOut{KwIn}{Input}
    \SetKwInOut{KwOut}{Output}
    \KwIn{Shape $S=(V,E)$}
    \KwOut{Growth process $\sigma$ for $S$ }
    \tcc{The algorithm consists of two sets of vertical and horizontal phases. A set of vertical phases denoted as $i_v$, and a set of horizontal phases denoted as $i_h$}
        $i_v$, $i_h \leftarrow 1$\\ 
        $\sigma$, $\sigma_v$, $\sigma_{h}=\{\}$ \\
         
        \While{ $R > 1$ }
        {
            ${Q}^v_{i_v} \leftarrow  \emptyset $\\
            \For {each $y=1$ to $R-1$ (i.e.,  iterating through rows from the bottom-most row to the top)}
            {
            
            \If {$y$ is odd }
                {
                
                \For{each $x = 1$ to $C$}
                {
                \If {
                $u_{xy}$, $u_{x(y+1)} \in V$ $\wedge$ 
                $u_{x(y+1)} \in C_x$}
                    {
                        ${Q}^v_{i_v}$ $\leftarrow $ ${Q}^v_{i_v} \cup \{\{(u_{xy}u_{x(y+1)}\},u_{x(y+1)})\}$\\
                        
                        $V \leftarrow V \setminus \{u_{x(y+1)}\}$\\
                        $E \leftarrow E \setminus \{u_{xy}u_{x(y+1)}\}$
                    }
                }
            }
        }
        
        $i_v \leftarrow i_v +1$\\
         
    }
    $\sigma_v \leftarrow \{{Q}^v_{i_v}, {Q}^v_{i_{v-1}},\ldots, {Q}^v_{i_1}\}$\\
    $\sigma \leftarrow \{\sigma_v\}$\\
     \tcc{Assuming we have done the previous vertical phases, we do the same for the columns $C$ to get the growth process $\sigma_h$ for the horizontal phases.}

        \While{ $C > 1$ }
        {
             $Q^h_{i_h} \leftarrow  \emptyset $\\
    
            \For {each $x=1$ to $C-1$ (i.e., iterating through columns from the leftmost column to the right) }
            {
            \If {$x$ is odd }
                {
               
                \If {
              $\exists v \in C_{j+1}$ such that $v\in N(u)$ (i.e., $v$ is a right neighbor of $u$)
                 $u_{x1}$, $u_{(x+1)1} \in V $
               }
                   {
                       $Q^h_{i_h}$ $\leftarrow $ $Q^h_{i_h} \cup \{\{(u_{xy}u_{(x+1)y}\},u_{(x+1)y})\}$\\
                        
                     $V \leftarrow V \setminus \{u_{(x+1)1}\}$\\
                    $E \leftarrow E \setminus \{u_{x1}u_{(x+1)1}\}$\\
               }
                
             }
       }

         $i_h \leftarrow i_h +1$
  }  
    $\sigma_h \leftarrow \{Q^h_{i_h}, Q^h_{i_{h-1}},\ldots, Q^h_{i_1}\}$\\
    
    $\sigma \leftarrow \sigma \cup \{ \sigma_h\}$ 
    
return $\sigma$\\
\end{algorithm}

\begin{theorem}\label{the:cbm-universal}
    Given any connected shape, $S$ with dimensions $l \times w$, Algorithm~\ref{alg:elimination-algorithm} 
    grows $S$ from a single node in $O(\log l + \log w)$ time steps. 
\end{theorem}

\begin{proof}
To prove the correctness of Algorithm~\ref{alg:elimination-algorithm} for growing any shape $S$ with $l$ rows and $w$ columns, for simplicity we will use induction focusing on one of the processes, say the horizontal $\sigma_h=\{Q^h_{i_h}, Q^h_{i_{h-1}},\ldots, Q^h_{i_1}\}$. The proof for the vertical phases would follow similarly due to the symmetry in how the algorithm handles rows and columns.
For the base case, since we use the process $\sigma_h=\{Q^h_{i_h}, Q^h_{i_{h-1}},\ldots, Q^h_{i_1}\}$, we consider the last phase of elimination, $i_h$, as the first phase of growth in the process $\sigma_h$. At this phase, the shape $S$ is reduced to its minimal form, consisting of a single row and column. Therefore, the shape at this phase $S_{i_{h}}$ corresponds to the single node $u_0$. The process $Q^h_{i_{h-1}}$ in the subsequent phase, $i_{h-1}$, which contains the nodes that were eliminated in this phase. Reversing this phase involves adding nodes adjacent to $u_0$, specifically, growing along the row and column that $u_0$ belongs to. This results in correctly constructing a sub-shape of $S$ that includes $u_0$ and its adjacent neighbor.
Assuming that up to the horizontal phase $i_{h-1}$, the algorithm has constructed a sub-shape $S_{i_h}$ from a single node which maintains the structure properties of $S$. In the following phase $i_h$, the process $Q^h_{i_h}$ specifies nodes to be added next. By induction, these nodes will be generated by the nodes in the sub-shape $S_{i_{h-1}}$. Therefore, the growth of the sub-shape $S_{i_{h}}$ from $S_{i_{h-1}}$ using the process $Q^h_{i_h}$ is correct and a substructure of $S$. Therefore, each horizontal phase in the reversed process $\sigma_h$ contributes correctly to the overall growth of the shape $S$.

The time complexity for the horizontal phases is $O(\log w)$ time steps due to the halving nature of the elimination process  across the columns. The growth during reconstruction is exponentially fast. This is because nodes that were eliminated in each phase are added back in parallel during the growth process. Similarly, for the vertical phases, the time complexity is $O(\log l)$ due to the halving of $l$ rows.
Therefore, the total number of time steps that combine the complexities of both vertical and horizontal phases is $O(\log l + \log w)$.
\end{proof}

\section{Adjacency vs. Connectivity Models – A Comparative Overview}\label{sec:comparative-overview}

\begin{definition}[Sparse Tree]\label{def:shapes-can-be-constructed-in-both-models}

    We define a tree $T$ as sparse if the adjacency closure $AC$ of $T$ equals the original tree $T$.
\end{definition}

\begin{observation}\label{obs:structures-constructed-in-both-models}
    By Definition~\ref{def:shapes-can-be-constructed-in-both-models}, if $T$ is a sparse tree $T=AC(T)$, then both models adjacency and connectivity can grow $T$ in $O(k \log n)$ time steps, see the intersecting area of Figure~\ref{fig:agm-vs-cgm}. Additionally, the adjacency graph model can be constructed by applying the algorithm of Theorem~\ref{theo:shape-spanning-tree}. 
\end{observation}

\begin{observation}\label{obs:structures--cannot-be-constructed-in-both-models}
    If a tree $T$ is not a sparse tree, meaning $T \neq AC(T)$, it implies that some nodes in $T$ are closely connected through non-tree edges. In such case, the tree $T$ falls into three scenarios as shown in Figure~\ref{fig:agm-vs-cgm}:
    \begin{enumerate}
        \item All close nodes lack edges between them. In this case, the connectivity model can grow it, but the adjacency model cannot (see CGM of Figure~\ref{fig:agm-vs-cgm}).
        \item All close nodes have edges. In this case, the adjacency model can grow it, but the connectivity model cannot (see AGM of Figure~\ref{fig:agm-vs-cgm}).
        \item Some nodes have edges while others do not. In this case, neither the connectivity nor the adjacency model can grow it.  An example of this class of shapes is the NICE shapes.
    \end{enumerate}
\end{observation}

\begin{figure}[ht]
\centering 
\includegraphics[width=2.7in]{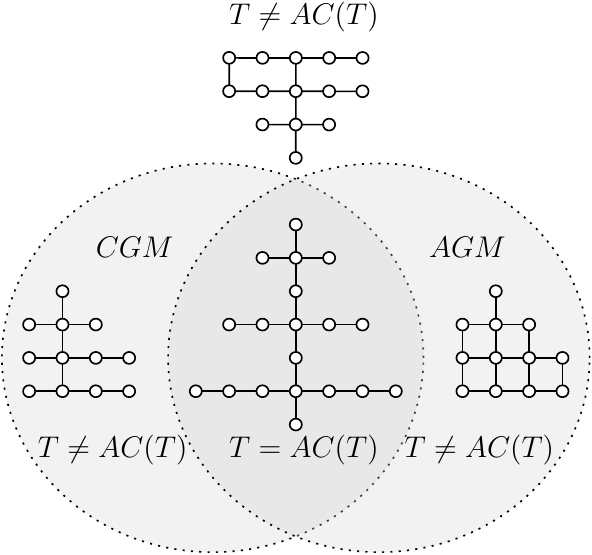}
\caption{An illustration of Observation~\ref{obs:structures-constructed-in-both-models}, where $T= AC(T)$ a non-empty intersection between the connectivity and adjacency graph models, and the three cases of $T \neq AC(T)$, Observation~\ref{obs:structures--cannot-be-constructed-in-both-models}.}
\label{fig:agm-vs-cgm}
\end{figure}

\begin{proposition}\label{prop:cgm}
    The connectivity graph model can grow $S$ iff $S$ is a tree. Also, if $S \in \mathcal{C}_k $ which means there exists a spanning tree $T$ of $S$ that has the $k$ turns property, then, the connectivity graph model can grow it in $O(k \log n)$ time steps.
\end{proposition}

\begin{proposition}\label{prop:agm}
    The adjacency graph model can grow $S$ iff $AC(S)=S$. Additionally, if $S \in \mathcal{C}_k $ which means there exists a spanning tree $T$ of $S$ that has the $k$ turns property, then, the adjacency graph model can grow it in $O(k \log n)$ time steps.
\end{proposition}

The following corollary holds for staircase structures, which can be grown in both models:

\begin{corollary}\label{coro:bounded-staircase}
    Any staircase shape $S$ with a bounded number of steps  
    can be grown in $O(\log n)$ time steps.
\end{corollary}

\section{Conclusion}\label{sec:conclusion} 
We explored the exponential growth of geometric structures for different combinations of growth operations and graph models. A combination that we did not study is cycle-preserving growth with neighbor handover. It would be interesting to know what is the class of shapes that can be grown in this model and if the additional property can be used to improve efficiency. Distributed versions of the growth processes studied in this paper is another direction to be investigated. Studying optimality and growth processes whose initial shape is not necessarily a single node are other rich directions of further research.

\bibliography{bibliography}

\end{document}